\documentclass[10pt,reqno]{article}
\usepackage{graphicx} 
\usepackage[utf8]{inputenc}
\usepackage{physics}
\usepackage[margin=2.5cm]{geometry}
\usepackage{mathtools,amsmath,amsfonts,amssymb,amsthm}
\usepackage{mathdots}
\usepackage{comment}
\usepackage{enumerate}
\usepackage{soul} 
\usepackage{bm}
\usepackage[dvipsnames]{xcolor}
\usepackage{hyperref}
\hypersetup{
    colorlinks=true,
    linkcolor=blue,
    citecolor=violet   
    }
\usepackage[nameinlink]{cleveref}
\Crefname{subsection}{Subsection}{subsections}
\allowdisplaybreaks

\newcommand\Fqt{\mathbb{F}_{q^t}}
\newcommand\Fqr{\mathbb{F}_{q^r}}
\newcommand\Fq{\mathbb{F}_q}

\renewcommand{\tilde}{\widetilde}
\renewcommand\epsilon{\varepsilon}
\newcommand\wt{\mathrm{wt}}

\newcommand\ie{\textit{i.e.,~}}
\newcommand\eg{\textit{e.g.,~}}
\theoremstyle{plain}
\newtheorem{theorem}{Theorem}[section]
\theoremstyle{plain}
\newtheorem{lemma}[theorem]{Lemma}
\theoremstyle{plain}

\theoremstyle{definition}
\newtheorem{definition}[theorem]{Definition}
\theoremstyle{definition}

\theoremstyle{definition}

\theoremstyle{plain}
\newtheorem{corollary}[theorem]{Corollary}
\theoremstyle{remark}
\newtheorem{remark}[theorem]{Remark}

\allowdisplaybreaks

\title{Generalized Skew Multivariate Goppa Codes}
\author{
Elena Berardini\thanks{CNRS; IMB, University of Bordeaux, France. Email: elena.berardini@math.u-bordeaux.fr}
\and
Pranav Trivedi\thanks{Department of Mathematics, UC Berkeley, Berkeley, CA. 
Email: pranavtrivedi@berkeley.edu}
}
\date{}

\begin{document}
\maketitle
\begin{abstract}
We introduce Generalized Skew Multivariate Goppa codes relying on the theory of multivariate Ore polynomials. These codes contain Generalized Skew Goppa codes as a special case. By providing a new parity-check matrix for the latter, we show that they are subfield subcodes of duals of Generalized Skew Reed--Solomon codes, up to the inverse field automorphism. This result is then used to study the parameters of Skew Multivariate Goppa codes, for which we provide bounds on their dimension and minimum distance.
\end{abstract}

\section{Introduction}
Among linear codes in the Hamming metric, Goppa codes \cite{Gopppa70,Goppa71} stand out for their many interesting properties: they can be efficiently decoded, their dimension can often be increased without decreasing the minimum distance, and they are one of the few linear codes for which the McEliece cryptosystem stays partially unbroken. Goppa codes have many different, though equivalent, characterizations,  all involving the use of a polynomial $g\in\Fq[x]$.  To name some, Goppa codes are alternant, that is, subfield subcodes of some Generalized Reed--Solomon codes. As such, they can also be defined through residues of differentials over the projective line. Furthermore, they also have a concrete description in terms of their parity-check matrix. 

Because of their features, Goppa codes and their generalizations have been studied extensively. In the Hamming metric, multivariate Goppa codes were introduced by replacing $g$ with a multivariate polynomial \cite{Lopez_2023}. Successively, with the rising interest in the rank metric, Skew and Generalized Skew Goppa codes were defined respectively in \cite{wang18,gomez-torrecillas2023} by replacing $g$ with an Ore polynomial. Finally, let us mention that Goppa codes in the sum-rank metric, called linearized Goppa codes, were introduced in \cite{CD23}, using skew residues. 
\smallskip

The first aim of this paper is to define Generalized Skew Multivariate Goppa codes and study their parameters, which are the multivariate version of Generalized Skew Goppa (GSG) codes, as introduced in \cite{gomez-torrecillas2023}. To this end, we also propose a new form of the parity-check matrix of GSG codes, which allows us to show that these codes are subfield subcodes of duals of Generalized Skew Reed--Solomon codes \cite{BU14}, up to the inverse field automorphism. This provides a characterization which is well-known for classical Goppa codes but was missing in the skew case. 

The paper is organized as follows. After some preliminaries on coding theory, in \Cref{sec:prelGoppa}, we recall Goppa codes in the Hamming metric and their multivariate version. In \Cref{sec:skewgoppa} we propose a systematic study of Generalized Skew Goppa codes and relate their construction with that of Generalized Skew Reed--Solomon codes. Finally, in \Cref{sec:skew.multivariate.goppa}, we exploit the theory of multivariate Ore polynomials recently developed in \cite{berardini24} to introduce the multivariate version of GSG codes, and study their parameters.

\section{Goppa codes in the Hamming metric and their variants}\label{sec:prelGoppa}
After recalling some preliminaries on coding theory, this section introduces the original construction of Goppa codes in the Hamming metric \cite{Gopppa70,Goppa71} and its multivariate version, introduced in \cite{Lopez_2023}. 

\subsection{Preliminaries}\label{sec:preliminaries}
Throughout the paper, we fix $\Fqt$ an extension of degree $t$ of the finite field $\Fq$, and $n$ a positive integer. A \emph{linear code} is a subspace $C \subseteq \Fqt^n$. We often say that $C$ is an $[n,k,d]$ linear code where $n$ is the block length, $k = \dim_{\Fqt} C$ and $d$ is the minimum Hamming distance defined as \[d = \min_{c \in C, c \neq 0} \wt(c)=\min_{c \in C, c \neq 0} \mid \{i\in\{1,\dots,n\}\,:\, c_i\neq 0\}\mid.\]
Given a subfield $\Fqr$ of $\Fqt$, one can consider $C\vert_{\Fqr} = C \cap \Fqr^n$ which is called a \emph{subfield subcode}. It is the restriction of $C$ to the codewords which have entries only in $\Fqr$.

Attached to a finite field extension, we also have the \emph{trace map} \[\Tr: \Fqt \to \Fqr\] given by \[\Tr(a) = \sum_{i=0}^{\frac{t}{r}-1} a^{(q^r)^i}\] which allows us to define the trace code.

\begin{definition}\label{def:trace_code}
    Given a linear code $C \subseteq \Fqt^n$, for $c = (c_1,\ldots, c_n)$, we define $\Tr(c) = (\Tr(c_1),\ldots, \Tr(c_n)) \in \Fqr^n$. We define the \emph{trace code} of $C$ to be $\Tr(C) = \{\Tr(c)\,:\, c \in C\}$.
\end{definition}

The following is a well-known result due to Delsarte, relating subfield subcodes and trace codes. For a proof, we refer the reader to \cite[Thm.~9.1.2]{stichtenoth}.
\begin{theorem}[Delsarte]\label{th:delsarte}
    For a code $C \subseteq \Fqt^n$, we have \[(C\vert_{\Fqr})^\perp = \Tr(C^\perp).\]
\end{theorem}

Finally, we recall the following known result on the minimum distance of the dual of the tensor product of codes, which we shall use in the rest of the paper. A proof of this result can be found, for instance, in
\cite[Prop.~2.14]{barraud2025dual}.
\begin{lemma}\label{lem:distance.dual.tensor.product}
    Let $C_i \subseteq \mathbb{F}^{n_i}$ be linear codes for $i = 1,2$ and let $d_i = d(C_i)$. Then \[(C_1^\perp\otimes C_2^\perp)^\perp = C_1 \otimes \mathbb{F}^{n_2} + \mathbb{F}^{n_1} \otimes C_2\] and its minimum distance is $\min(d_1,d_2)$.
\end{lemma}

We now proceed to describe Goppa codes and their nature as subfield subcodes of particular Generalized Reed--Solomon codes.

\subsection{Goppa codes}\label{sec:Goppa}

For $k$ a positive integer, a Generalized Reed--Solomon (GRS) code over $\Fqt$ is defined as \begin{equation*}
    \mathrm{GRS}_k(S, v) := \{(v_1\cdot f(s_1),\ldots, v_n\cdot f(s_n)): f \in \Fqt[X]_{< k}\},
\end{equation*} where $S = \{s_1,\ldots, s_n\}$ is a set of distinct elements of $\Fqt$ and $v = \{v_1,\ldots, v_n\}$ a set of not necessarily distinct elements of $\Fqt^*$. When $v_i=1$ for all $i$, we recover classical Reed--Solomon codes. Taking $k\leq n$, Generalized Reed--Solomon codes have parameters $[n,k,n+1-k]$, thus attaining the Singleton bound.

Let $g \in \Fqt[X]$ be a polynomial that does not vanish at any entries in $S$. Throughout the paper, we consider an integer $r$ such that $r\mid t$, so that $\Fqr$ is a subfield of $\Fqt$. The classical Goppa code associated to $(S, g, \Fqr)$ \cite[Ch.~12,\S~3]{MacWilliams_Sloane_1977} is defined as \begin{equation*}
    \Gamma(S, g, \Fqr) := \left\{c = (c_1,\ldots, c_n) \in \Fqr^n \, : \, \sum_{i=1}^n \frac{c_i}{X - s_i}\equiv 0 \mod g\right\}.
\end{equation*}
Goppa codes have length $n$, dimension $k\geq n-\deg g\cdot (t/r)$ and minimum distance $d\geq \deg g +1$.

When in the definition of GRS codes we take $v_i=g(s_i)^{-1}$ for some polynomial $g\in\Fqt[X]$ of degree $k$ which does not vanish at any of the $s_i$'s, we obtain so-called GRS codes via a Goppa code:
\begin{equation}
    \mathrm{GRS}_{\deg g} (S, g) := \{(g(s_1)^{-1}f(s_1),\ldots, g(s_n)^{-1}f(s_n)): f \in \Fqt[X]_{< \deg g}\}.
\end{equation} 

It is well known that the generator matrix for $\mathrm{GRS}_{\deg g}(S,g)$ is \begin{align*}G &= \begin{bmatrix}
    g(s_1)^{-1} & g(s_2)^{-1} & \cdots & g(s_n)^{-1}\\
    g(s_1)^{-1}s_1 & g(s_2)^{-1}s_2 & \cdots & g(s_n)^{-1}s_n\\
    \vdots & \vdots & \ddots & \vdots\\
    g(s_1)^{-1}s_1^{\deg g - 1} &g(s_2)^{-1}s_2^{\deg g - 1} &\cdots &g(s_n)^{-1}s_n^{\deg g - 1}\\
\end{bmatrix}\\ &= \begin{bmatrix}
    1 & 1 & \cdots & 1\\
    s_1 & s_2 & \cdots & s_n\\
    \vdots & \vdots & \ddots & \vdots\\
    s_1^{\deg g - 1} &s_2^{\deg g - 1} &\cdots &s_n^{\deg g - 1}\end{bmatrix}\begin{bmatrix}
    g(s_1)^{-1} & 0 & \cdots & 0\\
    0 & g(s_2)^{-1}  
    & \cdots & 0\\
    \vdots & \vdots & \ddots & \vdots\\
   0 & 0 &\cdots & g(s_n)^{-1}\\
\end{bmatrix}\end{align*} and it is a parity-check matrix for $\Gamma(S,g, \Fqr)$ \cite[p.~340]{MacWilliams_Sloane_1977}. Moreover, the dual of a GRS is another GRS. In the above case, the dual of $\mathrm{GRS}_{\deg g}(S,g)$ is given by $\mathrm{GRS}_{n-\deg g}(S, y)$ where for $i = 1,\ldots, n$ $$y_i = \frac{g(s_i)}{\prod_{j\neq i}(s_i - s_j)}.$$ The generator matrix $G$ is a parity-check matrix for $\mathrm{GRS}_{n-\deg g}(S,y)$ so naturally $\Gamma(S,g,\Fqr)$ is a subfield subcode of $\mathrm{GRS}_{n-\deg g}(S,y)$. We note that the dual of a GRS code via a Goppa code is not necessarily a GRS code via a Goppa code. 

\subsection{Multivariate Goppa codes}\label{sec:multivariate.goppa}
The construction of Goppa codes was extended to the multivariate case in \cite{Lopez_2023}. We recall here their construction and offer an alternative, shorter proof of the parameters of the so-called multivariate Goppa codes. The approach we take here will then be of inspiration for the study of the skew version of multivariate Goppa codes we shall perform in \Cref{sec:skew.multivariate.goppa}.

\begin{definition}\label{def:multgoppa}
Fix nonempty subsets $S_1,\ldots, S_m \subseteq \Fqt$. Let $\mathcal{S} := S_1 \times \cdots \times S_m \subseteq \Fqt^m.$ Enumerate the elements of $\mathcal{S} = \{\bm{s}_1,\ldots, \bm{s}_n\}.$ Choose $g \in \Fqt[\bm{x}] = \Fqt[x_1,\ldots, x_m]$ such that $g(\bm{s}_i) \neq 0$ for all $i$. Further, we assume $g = g_1\cdots g_m$ where $g_i \in \Fqt[x_i]$ and define $\rho = \prod_{i=1}^m \deg g_i$. Let $r\mid t$ be an integer. The multivariate Goppa code is defined as \begin{equation}\label{eq:multGoppa}
    \Gamma(\mathcal{S}, g,\Fqr):= \left\{(c_1,\ldots, c_n) \in \Fqr^n\, : \, \sum_{i=1}^n \frac{c_i}{\prod_{j= 1}^m (x_j-s_{ij})}\equiv 0 \mod g\right\},
\end{equation} where $\bm{s}_i = (s_{i1},\ldots, s_{im}) \in \mathcal{S}$. 
\end{definition}

We will often denote $\Gamma(\mathcal{S}, g, \Fqr)$ by $\Gamma(\mathcal{S},g)$ when the subfield is understood from the context. Notice, when $m = 1$, $\Gamma(\mathcal{S},g)$ reduces to a classical Goppa code. In general, it was shown in \cite[Thm.~8]{Lopez_2023} that the generator matrix for $$\bigotimes_{i = 1}^m \mathrm{GRS}_{\deg g_i}(S_i, g_i)=:T(\mathcal{S}, g)$$ is a parity-check matrix for $\Gamma(\mathcal{S}, g)$. In other words, $\Gamma(\mathcal{S}, g)$ is a subfield subcode of the dual of the code denoted $T(\mathcal{S}, g)$, that is, $\Gamma(\mathcal{S}, g)=(T(\mathcal{S}, g)^\perp\cap \Fqr^n)$. We point out that the dual of $T(\mathcal{S},g)$ can be expressed as follows
\begin{align*}
    (T(\mathcal{S},g))^\perp &= \left(\bigotimes_{i = 1}^m \mathrm{GRS}_{\deg g_i}(S_i, g_i)\right)^\perp\\
    &=\sum_{i = 1}^m \Fqt^{n_1} \otimes \cdots \otimes \mathrm{GRS}_{\deg g_i}(S_i, g_i)^\perp \otimes \cdots \otimes \Fqt^{n_m}\\
    &=\sum_{i = 1}^m \Fqt^{n_1} \otimes \cdots \otimes \mathrm{GRS}_{n_i-\deg g_i}(S_i, y_i) \otimes \cdots \otimes \Fqt^{n_m},
\end{align*}
where $n_i = |S_i|$. The chain of equalities follows from observing and inducting on the fact that $(C_1 \otimes C_2)^\perp = C_1^\perp \otimes \Fqt^{n_2} + \Fqt^{n_1} \otimes C_2^\perp$ and then applying the fact that the dual of a GRS code is another GRS code.

We know $\Gamma(\mathcal{S},g)$ is the subfield subcode of $(T(\mathcal{S},g))^\perp$ so we obtain the following result, which matches \cite[Cor.~15]{Lopez_2023}.

\begin{theorem}\label{thm:multivariate.goppa.code.parameters}
    The multivariate Goppa code $\Gamma(\mathcal{S},g)$ has parameters:
    \begin{itemize}
        \item Length $n = |\mathcal{S}|$,
        \item Dimension $k$ satisfying $n - \frac{t}{r}\rho\leq k \leq n -\rho$,
        \item Minimum distance $d \geq \min_{i\in [m]}\{\deg g_i + 1\}$.
    \end{itemize}
\end{theorem}

\begin{proof}
    The length is obvious. Let $k = \dim \Gamma(\mathcal{S}, g)$. Since $T(\mathcal{S},g)$ has dimension $\rho$, its dual $(T(\mathcal{S},g))^\perp$ has dimension $n - \rho$ and $\Gamma(\mathcal{S},g)$ is a subfield subcode of $(T(\mathcal{S},g))^\perp$ so $\dim \Gamma(\mathcal{S},g) \leq n - \rho$. By Delsarte's Theorem (\Cref{th:delsarte}), we have that for a code $C$ over $\Fqt$, \begin{equation*}
        \left(C \cap \Fqr^n\right)^\perp = \Tr(C^\perp),
    \end{equation*} and $\Tr: C \to \Tr(C)$ is a surjective $\Fqr$ linear mapping so the dimension of $C$ regarded as a $\Fqr$ vector space is $\frac{t}{r}\cdot \dim C$. Applying this result to \[\Gamma(\mathcal{S},g) = (\Tr(T(\mathcal{S},g)))^\perp\] we have \[\dim \Tr(T(\mathcal{S},g)) \leq \frac{t}{r}\rho\] from which we obtain  \[\dim \Gamma(\mathcal{S},g) \geq n - \frac{t}{r}\rho.\]
    Finally, by \Cref{lem:distance.dual.tensor.product}, the minimum distance of $(T(\mathcal{S},g))^\perp$ is $\min_{j\in [m]}\{\deg g_j + 1\}$ and the bound on the minimum distance of $\Gamma(\mathcal{S},g)$ follows.
\end{proof}

In \cite{Lopez_2023}, the authors introduced Augmented Cartesian codes to prove the minimum distance of $\Gamma(\mathcal{S},g)$. We showed that our proof of dimension and distance of $\Gamma(\mathcal{S},g)$ follows directly from the tensor product of GRS codes characterization and manipulating the dual of tensor product codes. We now show how our method leads to an alternative distance proof for Augmented Cartesian codes. We recall that for $g$ and $S$ as above, Augmented Cartesian codes are defined as
\begin{equation*}
ACar(\mathcal{S}, g)=\left\{\left(\frac{g}{L}(s_i)f(s_i)\right)_i \,:\,  s_i\in S, f\in L(A_g)\right\},
\end{equation*}
where $A_g=\prod_{j=1}^m \{0,\dots,n_j-1\}\setminus\prod_{j=1}^m \{n_j-\deg_{x_j}(g),\dots,n_j-1\}$ and $L(A_g)=\mathrm{Span}_{\Fqt} \{\bm{x}^a\,:\, a\in A_g\}$. Note that by \cite[Thm.~12]{Lopez_2023}, we have $$ACar(\mathcal{S},g) = (T(\mathcal{S},g))^\perp,$$ hence $\Gamma(\mathcal{S}, g)$ is the subfield subcode of $ACar(\mathcal{S}, g)$, that is, $\Gamma(\mathcal{S}, g)=ACar(\mathcal{S},g)\cap \Fqr^n$. Since \[ACar(\mathcal{S}, g) = \sum_{j = 1}^m \Fqt^{n_1} \otimes \cdots \otimes \mathrm{GRS}_{n_j-\deg g_j}(S_j, y_j) \otimes \cdots \otimes \Fqt^{n_m},\] it follows from \Cref{lem:distance.dual.tensor.product} that the minimum distance is $\min_{j\in [m]}\{\deg g_j + 1\}$. Indeed, this gives rise to a proof of the minimum distance of $ACar(\mathcal{S}, g)$, which differs from the one presented in \cite[Lem.~6]{Lopez_2023}.

\section{Generalized Skew Goppa codes}\label{sec:skewgoppa}
In this section, we generalize the construction of Goppa codes to the skew case, following the construction of Generalized Skew Goppa (GSG) codes from \cite{gomez-torrecillas2023}. In particular, in \Cref{sec:skewgoppa.parity.check}, we give a new description of the parity-check matrix of GSG codes, which will allow us to show in \Cref{sec:GSG.inside.GSRS} that GSG codes are subfield subcodes of the dual of so-called Generalized Skew Reed--Solomon codes (introduced in \cite{BU14} and recalled later in this section). We shall use this result in \Cref{sec:skew.multivariate.goppa} to study the multivariate version of GSG codes.
\smallskip

Let $\theta$ be the Frobenius automorphism for the finite field $\Fqt$ and let $R = \Fqt[X;\theta]$. This is an Ore polynomial ring in one variable $X$ over the finite field $\Fqt$. On this ring, we have the usual sum while the multiplication is twisted by the rule $X\cdot a =\theta(a) \cdot X$ for all $a\in\Fqt$. It is well-known that $R$ is a left and right Euclidean domain, and has left and right divisors and multiples. We refer to \cite{Ore33} for the theory of Ore polynomial rings.

Distinctness and linear independence do not carry over in the noncommutative setting of Ore polynomial rings. For an equivalent notion, the concept of P-independence was introduced. In this section, we present the results necessary for this work, and refer the reader to \cite{Lam_1986, Lam1988, LAM200443, delenclosleroy2007, gomez-torrecillas2023} for a more comprehensive treatment.

\begin{definition}
    A set of $n$ points $S = \{s_1,\ldots, s_n\} \subseteq \Fqt$ is said to be $\theta$-P-independent if the degree of the least common left multiple of $\{X - s_i\mid 1 \leq i \leq n\}$ in $\Fqt[X;\theta]$ is $n$. When the automorphism is clear, we simply say P-independent.
\end{definition}

A polynomial $0 \neq g \in R$ is called \emph{invariant} if $gR = Rg$. If $g\in R$ is an invariant polynomial, then $R/Rg$ is a ring.

Suppose $gh = h'g$ for some $h,h' \in R$. If $fh = g$ then $gh = h'fh$ so $g = h'f$. On the other hand, if $g = h'f$ then $h'fh = h'g$ so $fh = g$. This shows that $f\mid_r g$ if, and only if, $f \mid_l g$, where $\mid_r$ and $\mid_l$ denote the right and left division, respectively. It follows that, given the monomial $X - s$ for $s \in \Fqt$, $(X-s, g)_r = 1$ if, and only if, $(X-s, g)_l = 1$, where $(\cdot,\cdot)_r$ and $(\cdot,\cdot)_l$ denote the right and left gcd, respectively.

If $(X-s,g)_r = 1$, then $X-s+Rg$ is a unit in the ring $R/Rg$. Its left and right inverses therefore coincide. Consequently, there is a unique $h\in R$ with $\deg h<\deg g$ such that
\[(X-s)h-1\in Rg.\]
Equivalently, the same polynomial also satisfies $h(X-s)-1\in Rg$. We use the displayed, right-inverse form because it is the one compatible with the symbol placement in the GSG parity-check equation below.

We are now ready to define Generalized Skew Goppa codes \cite[Def.~1]{gomez-torrecillas2023}.
\begin{definition}\label{def:skew.goppa.code}
Let $0 \neq g \in R$ be invariant. Let $S=\{s_1,\ldots,s_n\} \subseteq \Fqt^*$ be $\theta$-P-independent elements such that $(X-s_j, g)_r = 1$ for $j = 1,\ldots,n$. Let $h_j \in R$ be the unique element such that $\deg h_j < \deg g$ and $(X-s_j)h_j - 1 \in Rg$. Let $\eta=\{\eta_1,\ldots,\eta_n\} \subseteq \Fqt^*$. Let $r\mid t$ be an integer. The Generalized Skew Goppa (GSG) code $\tilde\Gamma\subseteq \Fqr^n$ is 
\begin{equation*}
\begin{aligned}
        \tilde\Gamma(S, \eta, g, \Fqr) &\coloneqq \left\{c=(c_1,\ldots, c_n)\in \Fqr^n \, : \, \sum_{j=1}^n h_j\eta_jc_j = 0\right\}\\
        &=\left\{c=(c_1,\ldots, c_n)\in \Fqr^n \, : \, \sum_{j=1}^n h_j\eta_jc_j \in Rg\right\}.
\end{aligned}
    \end{equation*}
\end{definition}
If $\eta_j = 1$ for $j = 1,\ldots, n$, then we call $\tilde\Gamma$ a Skew Goppa code. It is isomorphic to the linearized Goppa codes introduced in \cite{wang18}. Additionally, if $\theta = \mathrm{id}$, $\tilde\Gamma$ recovers the classical Goppa code.

We now proceed to a study of GSG codes and their parity-check matrix. From \cite[Thm.~2]{wang18}, we know that linearized Goppa codes have dimension $k\geq n - \frac{t}{r}\deg g$ and minimum distance $d\geq \deg g + 1$. A decoding algorithm for GSG codes was provided in \cite{gomez-torrecillas2023}, from which one can deduce that their minimum distance is at least $\deg g+1$. In our coming study, we will have a slightly different approach than the aforementioned papers. In particular, we will make explicit the form of invariant polynomials in $R$, and successively rewrite the parity-check matrix of GSG codes in a way which will be convenient for us, as explained at the beginning of this section. In order to do so, we first need to recall some results on partial norms and P-independent sets.

\subsection{Norms and P-independent sets}\label{sec:norms}
To proceed in analyzing the parity-check matrix of the GSG code, the following results about partial norms will be useful. Put $\sigma=\theta^{-1}$. For $\tau\in\{\theta,\sigma\}$ and $i\in\mathbb N$, define
\[N_0^\tau(a)=1,\qquad N_i^\tau(a)=\prod_{s=0}^{i-1}\tau^s(a)\quad(i>0).\]
We write $N_i=N_i^\theta$. One can easily check that $N_{i+1}^\tau(a)=\tau(N_i^\tau(a))a=\tau^i(a)N_i^\tau(a)$. Note that $N_t^\tau$ is the classical field norm and, if $\tau=\mathrm{id}$, then $N_i^\tau(a)=a^i$.

These norms often show up in the context of polynomial evaluation and, consequently, Vandermonde matrices (see \eg \cite{BU14}).
Consider a skew polynomial $g(X) = g_\rho X^\rho + g_{\rho-1}X^{\rho-1} + \cdots + g_0$, then the right evaluation of $g$ at a point $a$ is well known \cite[Lem. 2.4]{LAM1988308} to be \[g(a) = \sum_{i=0}^\rho g_i N_i(a).\]
We recall the following two results from \cite[Prop.~8 and 9]{gomez-torrecillas2023}, which we state without proof. For $a\in \Fqt,b\in\Fqt^*$, we write $^ba=\theta(b)ab^{-1}$ and denote by $[a]_\theta=\{^ba\mid b\in\Fqt^*\}$ the conjugacy class of $a$ under $\theta$. The first lemma relates the notion of P-independent subsets of a conjugacy class to the linear independence of the element by which we are conjugating.
\begin{lemma}\label{lem:p_independence_conjugacy}
        Given $a \in \Fqt^*$, the P-independent subsets of $[a]_\theta$ are those of the form $\{\prescript{b_1}{}{a},\ldots, \prescript{b_m}{}{a}\}$ where $m \leq t$ and $\{b_1,\ldots, b_m\}$ is an $\Fq$-independent set. Additionally, $m = t$ if, and only if, the least common left multiple of $\{X-\prescript{b_i}{}{a}\mid 1 \leq i \leq m\}$ is $X^t - N_t(a)$.
    \end{lemma}

    The next lemma describes how P-independent sets can be partitioned by norms.

    \begin{lemma}
        A subset $A \subseteq \Fqt^*$ is P-independent if, and only if, $A = A_1 \cup \cdots \cup A_s$, where $A_i \subseteq [a_i]_\theta$ is P-independent for all $i = 1,\ldots, s$, and $a_1,\ldots, a_s \in \Fqt^*$ have distinct norms.
    \end{lemma}

The inversion map does not generally preserve P-independence for the same automorphism. The relevant operation here, however, changes the automorphism at the same time.

\begin{lemma}\label{lem:inverse.automorphism.p.independence}
Let $S\subseteq\Fqt^*$. Then $S$ is $\theta$-P-independent if, and only if, $S^{-1}=\{s^{-1}:s\in S\}$ is $\sigma$-P-independent, where $\sigma=\theta^{-1}$.
\end{lemma}
\begin{proof}
It is enough to argue on each conjugacy class. Write $s_i=\theta(b_i)ab_i^{-1}$. By \Cref{lem:p_independence_conjugacy}, the elements $s_i$ are $\theta$-P-independent precisely when the $b_i$ are $\Fq$-linearly independent. Moreover,
\[s_i^{-1}=b_i a^{-1}\theta(b_i)^{-1}=\sigma(\theta(b_i))a^{-1}\theta(b_i)^{-1}.\]
Thus $s_i^{-1}$ is the $\sigma$-conjugate of $a^{-1}$ by $\theta(b_i)$. Since $\theta$ preserves $\Fq$-linear independence, the claim follows on each conjugacy class. Distinct classes remain distinct because their field norms are inverted. The converse follows by interchanging $\theta$ and $\sigma$.
\end{proof}

No additional hypothesis on P-independence after inversion is needed: the inverse points naturally belong to the Ore ring governed by $\sigma$, not by $\theta$.

\subsection{Parity-check matrix of GSG codes}\label{sec:skewgoppa.parity.check}
A parity-check matrix for the Generalized Skew Goppa codes was studied in \cite[Sec.~4]{gomez-torrecillas2023}. In this subsection, we rewrite it in a way which will be convenient for us later on in \Cref{sec:skew.multivariate.goppa}.

To this end, it will be convenient to characterize elements $g \in R$ that are invariant as required in \Cref{def:skew.goppa.code}. A more general version of \Cref{thm:invariant.elements.of.R} below is proven in \cite[Thm~1.1.22]{jacobson2009finite}, but we provide a self-contained characterization for completeness. In the following, we will use the well-known result that the center of $R$, that is, the set of all elements of $R$ that commute with every element of $R$, is $Z(R)=\Fq[X^t]$.

\begin{theorem}\label{thm:invariant.elements.of.R}
Let $R = \mathbb{F}_{q^t}[X; \theta]$.
Then a nonzero element $g \in R$ satisfies $gR = Rg$ 
(\ie $g$ is \emph{invariant}) 
if, and only if,
\[
g = a\cdot v \cdot X^l,
\]
for some $v(X) \in Z(R)$, $a \in \mathbb{F}_{q^t}$, and $l \in \mathbb{N}$.
\end{theorem}
\begin{proof}
    Clearly $0 \in \Fqt$ is invariant. Any nonzero element $a \in \Fqt$ is invariant because for any $f = \sum_{i=0}^k f_i X^i \in R$, \[Ra \ni f \cdot a = \sum_{i=0}^k f_i \theta^i(a) X^i = a \cdot \sum_{i=0}^k f_i \frac{\theta^i(a)}{a} X^i \in aR.\]
    For any $l \in \mathbb{N}$, $X^l$ is invariant because \[RX^l \ni f \cdot X^l = \sum_{i=0}^k f_iX^{i+l} = X^l \sum_{i=0}^k \theta^{-l}(f_i)X^i \in X^lR.\]
    Hence, any invariant element of $R$ has the form $a \cdot v(X) \cdot X^j$ where $a \in \Fqt$ and \[v(X) = 1 + v_1X + \cdots + v_mX^m,\quad v_i \in \Fqt, v_m \neq 0\] is invariant. For $v(X)$ to be invariant, we must satisfy that for every $b \in \Fqt$ there exists $b' \in R$ such that $v \cdot b = b' \cdot v$ and there exists $r(X) \in R$ such that $v(X) \cdot X = r(X) \cdot v(X)$. By degree considerations, it follows that $b' \in \Fqt$ and $\deg r = 1$ so $r(X) = r_1X + r_0$ for $r_0, r_1 \in \Fqt$.

    Suppose $b \neq 0$ (otherwise, invariance is trivially satisfied). From $v \cdot b =b' \cdot v$, we have \begin{align*}
        v\cdot b &= (1 + v_1X + \cdots + v_mX^m)b\\
        &= b + v_1 \theta(b)X + \cdots + v_m\theta^m(b) X^m\\
        &=b' + b'v_1X + \cdots + b'v_mX^m\\
        &= b'\cdot v
    \end{align*} which implies $b = b'$ and $\theta^i(b) = b' = b$ so $t \mid i$. Furthermore, because $v(X) \cdot X = r(X) \cdot v(X)$, we have \begin{align*}
        v(X) \cdot X &= X + v_1X^2 + \cdots + v_mX^{m+1}\\
        &= r_0 + (r_1 + r_0v_1)X + \cdots r_1\theta(v_m)X^{m+1}\\
        &= (r_1 X + r_0)(1 + v_1X + \cdots + v_mX^m)\\
        &= r(X) \cdot v(X)
    \end{align*} which forces $r_0 = 0$, $r_1 = 1$, and $\theta(v_i) = v_i$ so $v_i \in \Fq$. Hence, $r(X) = X$ and $v(X) \in \Fq[X^t] = Z(R)$, as desired.
\end{proof}

We are now ready to construct the parity-check matrix of GSG codes. This time we divide with the divisor on the left. Let
\[g(X)=a\,v(X)X^l=a(v_0+v_1X^t+\cdots+v_mX^{tm})X^l\]
be as in \Cref{thm:invariant.elements.of.R}, where $a,v_m\neq0$, $v_0=1$, and $\rho=tm+l$. Write
\[g(X)=\sum_{k=0}^{\rho}g_kX^k,\qquad
g_u=\begin{cases}av_{m-b}&u=\rho-bt,\quad b=0,\ldots,m,\\0&\text{otherwise.}\end{cases}\]
For every $j$, perform the left division
\begin{equation}\label{eq:left.division.g}
g(X)=(X-s_j)q_j(X)+r_j,\qquad q_j(X)=\sum_{i=0}^{\rho-1}q_{i,j}X^i,\quad r_j\in\Fqt.
\end{equation}
Since
\[(X-s_j)q_j=\sum_{i=0}^{\rho-1}\theta(q_{i,j})X^{i+1}-\sum_{i=0}^{\rho-1}s_jq_{i,j}X^i,\]
comparison of coefficients gives
\[q_{i,j}=\sigma(g_{i+1})+\sigma(s_j)\sigma(q_{i+1,j}),\qquad 0\leq i\leq\rho-1,\]
where $q_{\rho,j}=0$. Iterating this recurrence yields
\begin{align}
q_{i,j}
&=\sum_{v=i+1}^{\rho}\left(\prod_{a=1}^{v-i-1}\sigma^a(s_j)\right)\sigma^{v-i}(g_v)\label{eq:h_structure_skew}\\
&=\sum_{b=0}^{\left\lfloor\frac{\rho-i-1}{t}\right\rfloor}
\left(\prod_{a=1}^{\rho-bt-i-1}\sigma^a(s_j)\right)
\sigma^{\rho-bt-i}(g_{\rho-bt}).\nonumber
\end{align}
The constant coefficient in \labelcref{eq:left.division.g} gives
\begin{equation}\label{eq:left.remainder}
r_j=g_0+s_jq_{0,j}=\sum_{i=0}^{\rho}N_i^\sigma(s_j)\sigma^i(g_i).
\end{equation}
The map
\[
\dagger:\Fqt[X;\theta]\longrightarrow\Fqt[X;\sigma],\qquad
\left(\sum_i f_iX^i\right)^\dagger
=\sum_iX^if_i=\sum_i\sigma^i(f_i)X^i,
\]
is an anti-isomorphism. Define the corresponding coefficient transform
\begin{equation}\label{eq:bar.g}
\bar g(X):=\sum_{i=0}^{\rho}\bar g_iX^i\in\Fqt[X;\sigma],
\qquad \bar g_i:=\sigma^i(g_i).
\end{equation}
Then $r_j=\bar g(s_j)\neq0$. Multiplying \labelcref{eq:left.division.g} on the right by $-r_j^{-1}$ gives
\[g(X)(-r_j^{-1})=(X-s_j)(-q_j(X)r_j^{-1})-1.\]
As $gR=Rg$, the desired inverse representative is therefore
\begin{equation}\label{eq:h.from.left.quotient}
h_j(X)=-q_j(X)r_j^{-1}.
\end{equation}

Write \(h_j=\sum_{i=0}^{\rho-1}h_{i,j}X^i\). Since
\[
h_j=\sum_{i=0}^{\rho-1}X^i\sigma^i(h_{i,j}),
\]
the equation \(\sum_jh_j\eta_jc_j=0\) is equivalent to
\[
\sum_{j=1}^n\sigma^i(h_{i,j})\eta_jc_j=0,
\qquad i=0,\ldots,\rho-1.
\]
By \labelcref{eq:h.from.left.quotient}, \(h_{i,j}=-q_{i,j}\theta^i(r_j^{-1})\). After multiplying every row by \(-1\), a parity-check matrix is therefore
\begin{equation}\label{eq:HADE}
H=ADE,\qquad
A_{i,j}=\sigma^i(q_{i,j}),\quad
D=\operatorname{diag}(r_1^{-1},\ldots,r_n^{-1}),\quad
E=\operatorname{diag}(\eta_1,\ldots,\eta_n).
\end{equation}
Applying \(\sigma^i\) to \labelcref{eq:h_structure_skew}, and using \(\bar g_k=\sigma^k(g_k)\), gives
\begin{equation}\label{eq:A.entries}
A_{i,j}=
\sum_{b=0}^{\left\lfloor(\rho-i-1)/t\right\rfloor}
\bar g_{\rho-bt}
\frac{N_{\rho-bt}^{\sigma}(s_j)}{N_{i+1}^{\sigma}(s_j)}.
\end{equation}

Set
\[
B_{i,j}:=\bar g_\rho
\frac{N_\rho^\sigma(s_j)}{N_{i+1}^\sigma(s_j)},
\qquad 0\leq i\leq\rho-1.
\]
Since \(\sigma^t=\mathrm{id}\),
\[
\frac{N_{\rho-bt}^{\sigma}(s_j)}{N_{i+1}^{\sigma}(s_j)}
=
\frac{N_{\rho}^{\sigma}(s_j)}{N_{i+bt+1}^{\sigma}(s_j)}.
\]
Consequently \(A=UB\), where \(U\) is the upper triangular matrix determined by
\[
U_{i,i+bt}=\frac{\bar g_{\rho-bt}}{\bar g_\rho}
\quad
\left(0\leq b\leq\left\lfloor\frac{\rho-i-1}{t}\right\rfloor\right),
\]
and all other entries are zero. In particular, \(U_{i,i}=1\), so \(U\) is invertible.

The entries of \(B\) factor as
\[
B_{i,j}
=
\sigma\!\left(N_i^\sigma(s_j^{-1})\right)
\bar g_\rho\,\sigma\!\left(N_{\rho-1}^\sigma(s_j)\right).
\]
Combining this identity with \(r_j=\bar g(s_j)\), we obtain
\begin{equation}\label{eq:corrected.parity.matrix}
U^{-1}H=
\begin{bmatrix}
1&\cdots&1\\
\sigma(N_1^\sigma(s_1^{-1}))&\cdots&\sigma(N_1^\sigma(s_n^{-1}))\\
\vdots&\ddots&\vdots\\
\sigma(N_{\rho-1}^\sigma(s_1^{-1}))&\cdots&
\sigma(N_{\rho-1}^\sigma(s_n^{-1}))
\end{bmatrix}
\operatorname{diag}(u_1,\ldots,u_n),
\end{equation}
where
\begin{equation}\label{eq:corrected.multipliers}
u_j=
\bar g_\rho\,\sigma\!\left(N_{\rho-1}^\sigma(s_j)\right)
\bar g(s_j)^{-1}\eta_j.
\end{equation}

\begin{theorem}\label{th:newparitycheck}
The matrix in \labelcref{eq:corrected.parity.matrix} is a parity-check matrix for
\(\tilde\Gamma(S,\eta,g,\Fqr)\). It is obtained from the coefficient
matrix in \labelcref{eq:HADE} by the invertible row operation \(U^{-1}\).
\end{theorem}

\begin{theorem}\label{th:rank.corrected.parity.check}
Assume \(\rho=\deg g\leq n\). The matrix \(H\) has rank \(\rho\), and
every set of \(\rho\) columns of \(H\) is linearly independent.
\end{theorem}
\begin{proof}
By \Cref{lem:inverse.automorphism.p.independence}, the
\(\theta\)-P-independence of \(S\) implies the \(\sigma\)-P-independence
of \(S^{-1}\). Hence every square \(\rho\times\rho\) minor obtained by
choosing \(\rho\) columns of the \(\sigma\)-Vandermonde matrix on
\(S^{-1}\) is nonsingular \cite[Thm.~10]{Lam_1986}. Applying \(\sigma\)
entrywise and multiplying columns by the nonzero numbers \(u_j\) do
not change this property. The result follows from
\labelcref{eq:corrected.parity.matrix}.
\end{proof}

\begin{corollary}\label{cor:GSG.parameters}
Let \(\tilde\Gamma(S,\eta,g,\Fqr)\) be a Generalized Skew Goppa code
with \(\deg g\leq n\). Then it has length \(n\), dimension \(k\)
satisfying
\[
n-\frac{t}{r}\deg g\leq k\leq n-\deg g,
\]
and minimum distance \(d\geq\deg g+1\).
\end{corollary}
\begin{proof}
The ambient \(\Fqt\)-linear kernel of \(H\) is an
\([n,n-\rho,\rho+1]\) MDS code by
\Cref{th:rank.corrected.parity.check}. Its subfield subcode therefore
has minimum distance at least \(\rho+1\) and dimension at most
\(n-\rho\). Expanding the \(\rho\) rows of \(H\) in an \(\Fqr\)-basis
of \(\Fqt\) gives at most \(\rho t/r\) parity checks over \(\Fqr\),
which proves the lower dimension bound.
\end{proof}

\begin{remark}\label{rem:commutative_h}
When \(\theta=\mathrm{id}\), one has \(\sigma=\mathrm{id}\),
\(\bar g=g\), and
\[
q_{i,j}=\sum_{k=i+1}^{\rho}g_ks_j^{k-i-1}.
\]
Thus \labelcref{eq:corrected.parity.matrix} reduces, up to invertible row
operations and nonzero column multipliers, to the usual Goppa
parity-check matrix \cite[p.~340]{MacWilliams_Sloane_1977}.
\end{remark}

\begin{remark}\label{rem:gomez_distance_comparison}
The decoding argument of \cite{gomez-torrecillas2023} proves the same
distance bound. The matrix proof above now gives it directly.
\end{remark}
\subsection{Generalized Skew Goppa codes as subfield subcodes}\label{sec:GSG.inside.GSRS}
Generalized Skew Evaluation (GSE) codes and Generalized Skew
Reed--Solomon (GSRS) codes have been studied in
\cite{BU14,liu2015,gomez-torrecillas-2019}. We keep the automorphism in
the notation because the parity-check matrix above naturally uses
\(\sigma=\theta^{-1}\).

\begin{definition}\label{def:gse.and.gsrs}
Let \(\tau\) be an automorphism of \(\Fqt\), let \(1\leq k\leq n\), let
\(S=\{s_1,\ldots,s_n\}\subseteq\Fqt\), and let
\(v=(v_1,\ldots,v_n)\in(\Fqt^*)^n\). Put
\[
V_k^\tau(S):=
\begin{bmatrix}
1&\cdots&1\\
N_1^\tau(s_1)&\cdots&N_1^\tau(s_n)\\
\vdots&\ddots&\vdots\\
N_{k-1}^\tau(s_1)&\cdots&N_{k-1}^\tau(s_n)
\end{bmatrix}.
\]
The Generalized Skew Evaluation code is
\[
\operatorname{GSE}_k^\tau(S,v)
=
\left\{(v_1f(s_1),\ldots,v_nf(s_n)):
f\in\Fqt[X;\tau],\ \deg f<k\right\}.
\]
If \(S\) is \(\tau\)-P-independent, this code is denoted
\(\operatorname{GSRS}_k^\tau(S,v)\).
\end{definition}

A generator matrix of \(\operatorname{GSRS}_k^\tau(S,v)\) is
\(V_k^\tau(S)\operatorname{diag}(v)\). As a nonzero skew polynomial of
degree less than \(k\) cannot vanish on \(k\) points of a
P-independent set, GSRS codes are \([n,k,n-k+1]\) MDS codes.

\begin{lemma}\label{lem:theta.commutes.with.perp}
For every field automorphism \(\tau\) and every linear code
\(C\subseteq\Fqt^n\),
\[
\tau(C^\perp)=\tau(C)^\perp.
\]
\end{lemma}
\begin{proof}
Apply \(\tau\) componentwise to the equations \(z\cdot c=0\), for
\(z\in C^\perp\) and \(c\in C\), and use dimension counting.
\end{proof}

\begin{theorem}\label{thm:GSG.subfield.subcode.GSRS}
Let \(\tilde\Gamma(S,\eta,g,\Fqr)\) be a Generalized Skew Goppa code,
let \(\rho=\deg g\leq n\), and define \(u=(u_1,\ldots,u_n)\) by
\labelcref{eq:corrected.multipliers}. Then
\begin{align*}
\tilde\Gamma(S,\eta,g,\Fqr)
&=
\left[
\sigma\!\left(
\operatorname{GSRS}_\rho^\sigma
(S^{-1},\sigma^{-1}(u))
\right)
\right]^\perp\cap\Fqr^n\\
&=
\sigma\!\left(
\operatorname{GSRS}_\rho^\sigma
(S^{-1},\sigma^{-1}(u))^\perp
\right)\cap\Fqr^n.
\end{align*}
\end{theorem}
\begin{proof}
By \Cref{lem:inverse.automorphism.p.independence}, \(S^{-1}\) is
\(\sigma\)-P-independent. Hence
\[
G=V_\rho^\sigma(S^{-1})
\operatorname{diag}(\sigma^{-1}(u_1),\ldots,\sigma^{-1}(u_n))
\]
generates
\(\operatorname{GSRS}_\rho^\sigma
(S^{-1},\sigma^{-1}(u))\).
\Cref{eq:corrected.parity.matrix} says precisely that
\(U^{-1}H=\sigma(G)\). Since \(U\) is invertible, \(H\) and
\(\sigma(G)\) have the same row space, and taking the kernel over
\(\Fqt\), followed by intersection with \(\Fqr^n\), proves the first
equality. The second follows from
\Cref{lem:theta.commutes.with.perp}.
\end{proof}

\begin{remark}
The ambient code in \Cref{thm:GSG.subfield.subcode.GSRS} is MDS.
Thus the theorem also gives the dimension and distance bounds in
\Cref{cor:GSG.parameters}. The change from \(\theta\) to
\(\sigma=\theta^{-1}\) is part of the anti-isomorphism that reverses
the side of multiplication; it is not an extra hypothesis on the
evaluation set.
\end{remark}
\section{Generalized Skew Multivariate Goppa codes}\label{sec:skew.multivariate.goppa}
We now extend the corrected one-variable construction to the
almost commutative multivariate Ore setting of \cite{berardini24}.

Consider the ring
\[
R=\Fqt[\mathbf X;\theta]=\Fqt[X_1,\ldots,X_m;\theta]
\]
whose multiplication is determined by
\[
X_iX_j=X_jX_i,\qquad
X_i a=aX_i\ (i<m),\qquad
X_ma=\theta(a)X_m
\]
for \(a\in\Fqt\). Put
\(
L=\mathbb N^{m-1}\times t\mathbb N.
\)
By \cite[Prop.~1.1]{berardini24}, the center is
\[
Z(R)=\Fq[\mathbf X^L].
\]
It follows from \Cref{thm:invariant.elements.of.R} that the invariant polynomials of $R$ are in $\Fqt\cdot Z \cdot \mathbf{X}^l$, where $l\in\mathbb{N}^m$. 

Fix nonempty sets \(S_1,\ldots,S_{m-1}\subseteq\Fqt\) and
\(S_m\subseteq\Fqt^*\). Assume that the
elements of \(S_i\) are distinct for \(i<m\), and that \(S_m\) is
\(\theta\)-P-independent. Set
\[
\mathcal S=S_1\times\cdots\times S_m,\qquad
n_i=|S_i|,\qquad n=\prod_{i=1}^m n_i.
\]
Write a point as
\(\mathbf s_{\mathbf j}=(s_{1j_1},\ldots,s_{mj_m})\), where
\(\mathbf j=(j_1,\ldots,j_m)\).

For \(i<m\), let \(g_i\in\Fqt[X_i]\) be invariant in \(R\), have
degree \(\rho_i\), and satisfy
\(g_i(s_{ij_i})\neq0\). Let
\(g_m\in\Fqt[X_m;\theta]\) be invariant of degree \(\rho_m\), and
assume \((X_m-s_{mj_m},g_m)_r=1\) for every \(j_m\). Put
\[
g=g_1\cdots g_m,\qquad \rho=\prod_{i=1}^m\rho_i.
\]
We assume throughout that \(\rho_i\leq n_i\) for all \(i\).

For every \(i<m\), use the ordinary division
\[
g_i=(X_i-s_{ij_i})q_{ij_i}+r_{ij_i},
\qquad r_{ij_i}=g_i(s_{ij_i}),
\]
and set \(h_{ij_i}=-q_{ij_i}r_{ij_i}^{-1}\). For the last variable,
use the left division
\[
g_m=(X_m-s_{mj_m})q_{mj_m}+r_{mj_m},
\qquad r_{mj_m}=\bar g_m(s_{mj_m}),
\]
where
\[
\bar g_m(X_m)=\sum_{k=0}^{\rho_m}
\sigma^k(g_{m,k})X_m^k\in\Fqt[X_m;\sigma].
\]
Set \(h_{mj_m}=-q_{mj_m}r_{mj_m}^{-1}\). Thus, for every \(i\),
\[
(X_i-s_{ij_i})h_{ij_i}-1\in Rg_i.
\]

The order of the factors now matters. Define
\begin{equation}\label{eq:multivariate.h.order}
h_{\mathbf j}:=
h_{mj_m}h_{m-1,j_{m-1}}\cdots h_{1j_1}.
\end{equation}
Placing the skew factor first ensures that the right coefficients of
\(h_{\mathbf j}\) are products of the ordinary coefficients of the
first \(m-1\) factors and the right coefficients of the last factor.
Choose \(\eta_1,\ldots,\eta_{n_m}\in\Fqt^*\) and put
\(\eta_{\mathbf j}=\eta_{j_m}\).

\begin{definition}\label{def:gsmg}
The Generalized Skew Multivariate Goppa (GSMG) code is
\begin{align*}
\tilde\Gamma(\mathcal S,\eta,g,\Fqr)
&=
\left\{(c_{\mathbf j})_{\mathbf j\in\mathcal S}\in\Fqr^{\mathcal S}:
\sum_{\mathbf j\in\mathcal S}
h_{\mathbf j}\eta_{\mathbf j}c_{\mathbf j}=0\right\}\\
&=
\left\{(c_{\mathbf j})_{\mathbf j\in\mathcal S}\in\Fqr^{\mathcal S}:
\sum_{\mathbf j\in\mathcal S}
h_{\mathbf j}\eta_{\mathbf j}c_{\mathbf j}\in Rg\right\}.
\end{align*}
\end{definition}

When \(m=1\), \Cref{def:gsmg} is
\Cref{def:skew.goppa.code}. When \(\theta=\mathrm{id}\) and all
multipliers are \(1\), the order in \labelcref{eq:multivariate.h.order}
becomes immaterial and the construction specializes to the
multivariate Goppa code of \Cref{def:multgoppa}.

\subsection{Parity-check matrix and parameters}

Write
\[
q_{ij_i}=\sum_{b=0}^{\rho_i-1}q_{ij_i,b}X_i^b.
\]
For \(i<m\), define
\[
P_i=
\left(q_{ij_i,b}r_{ij_i}^{-1}\right)_{
0\leq b<\rho_i,\ 1\leq j_i\leq n_i}.
\]
For the skew factor define its right-coefficient matrix
\[
P_m=
\left(\sigma^b(q_{mj_m,b})r_{mj_m}^{-1}\eta_{j_m}\right)_{
0\leq b<\rho_m,\ 1\leq j_m\leq n_m}.
\]
The global signs coming from the \(h_{ij_i}\)'s do not affect the
kernel. By \labelcref{eq:multivariate.h.order}, a parity-check matrix of the
GSMG code is
\begin{equation}\label{eq:gsmg.corrected.parity}
H_{\mathrm{multi}}=
P_1\otimes\cdots\otimes P_{m-1}\otimes P_m.
\end{equation}
Indeed, moving the coefficients of the \(X_m\)-powers to the right
produces the factors \(\sigma^b(q_{mj_m,b})\), while the coefficients
of the commuting variables are unchanged.

The usual Goppa row reduction gives
\[
\operatorname{rowsp}(P_i)=
\operatorname{GRS}_{\rho_i}(S_i,g_i)
\qquad(i<m).
\]
For the last factor, set
\[
u_{j_m}^{(m)}
=
\bar g_{m,\rho_m}
\sigma\!\left(N_{\rho_m-1}^\sigma(s_{mj_m})\right)
\bar g_m(s_{mj_m})^{-1}\eta_{j_m}.
\]
The univariate calculation in \labelcref{eq:corrected.parity.matrix} gives
\[
\operatorname{rowsp}(P_m)
=
\sigma\!\left(
\operatorname{GSRS}_{\rho_m}^\sigma
(S_m^{-1},\sigma^{-1}(u^{(m)}))
\right).
\]
Here \(S_m^{-1}\) is \(\sigma\)-P-independent automatically by
\Cref{lem:inverse.automorphism.p.independence}.

Define
\begin{equation}\label{eq:corrected.tensor.code}
\widetilde T(\mathcal S,g):=
\left(\bigotimes_{i=1}^{m-1}
\operatorname{GRS}_{\rho_i}(S_i,g_i)\right)
\otimes
\sigma\!\left(
\operatorname{GSRS}_{\rho_m}^\sigma
(S_m^{-1},\sigma^{-1}(u^{(m)}))
\right).
\end{equation}
\Cref{eq:gsmg.corrected.parity,eq:corrected.tensor.code}
give the subfield-subcode characterization
\begin{equation}\label{eq:gsmg.subfield.subcode}
\tilde\Gamma(\mathcal S,\eta,g,\Fqr)
=
\widetilde T(\mathcal S,g)^\perp\cap\Fqr^n.
\end{equation}

\begin{theorem}\label{thm:gsmg.code.parameters}
The Generalized Skew Multivariate Goppa code
\(\tilde\Gamma(\mathcal S,\eta,g,\Fqr)\) has length \(n=|\mathcal S|\),
dimension \(k\) satisfying
\[
n-\frac{t}{r}\rho\leq k\leq n-\rho,
\]
and minimum distance
\[
d\geq\min_{i\in[m]}\{\deg g_i+1\}.
\]
\end{theorem}
\begin{proof}
Every factor in \labelcref{eq:corrected.tensor.code} is MDS of dimension
\(\rho_i\), so \(\widetilde T(\mathcal S,g)\) has dimension
\(\prod_i\rho_i=\rho\). \Cref{eq:gsmg.subfield.subcode} gives the upper dimension bound. By
Delsarte's theorem,
\[
\tilde\Gamma(\mathcal S,\eta,g,\Fqr)^\perp
=
\operatorname{Tr}(\widetilde T(\mathcal S,g)),
\]
whose dimension is at most \((t/r)\rho\); this proves the lower
dimension bound. Finally, the duals of all tensor factors are MDS with
minimum distances \(\rho_i+1\). Applying
\Cref{lem:distance.dual.tensor.product} inductively gives
\[
d(\widetilde T(\mathcal S,g)^\perp)
=
\min_{i\in[m]}(\rho_i+1).
\]
The same value is therefore a lower bound for its subfield subcode.
\end{proof}

\section{Conclusion}
In this paper, we constructed a multivariate version of Skew Goppa codes by using the multivariate Ore polynomial ring $\Fqt[X_1,\ldots, X_m;\theta]$, where the only noncommutative variable is the last one. For the sake of completeness, one could study a more general multivariate construction from the ring $\Fqt[X_1,\ldots, X_m;\theta_1,\dots,\theta_m]$. However, note that even for the linearized Reed--Muller codes of \cite{berardini24}, the ``almost commutative" case is the one giving the best parameters. 

We also showed that Generalized Skew Goppa codes are subfield subcodes of duals of Generalized Skew Reed--Solomon codes, up to the inverse automorphism and the Generalized Skew Multivariate Goppa codes are subfield subcodes of the dual of tensor products of Generalized Reed--Solomon codes and a Generalized Skew Reed--Solomon code, up to the inverse automorphism.

Finally, we believe it would be interesting to study the relation between GSG codes and the linearized Goppa codes introduced in \cite{CD23} using skew residues. 

\bibliographystyle{alpha}
\bibliography{refs}
        
\end{document}